\documentclass[12pt]{article}
\usepackage{amssymb,amsmath,dsfont}
\usepackage{amsthm}

\def\newtheorems{\newtheorem{theorem}{Theorem}[section]

\newtheorem{lemma}[theorem]{Lemma}

\newtheorem{definition}[theorem]{Definition}}

\newtheorems

\newcommand{\THEN}{{\sf then}}

\newcommand{\ELSE}{{\sf else}}

\newcommand{\elseif}{{\sf elseif}}

\newcommand{\Read}{\mbox{\it read}}

\newcommand{\IF}{{\sf if}}

\newcommand{\END}{{\sf end}}
\newcommand{\then}{\THEN}

\newcommand{\Readit}{\mbox{\it read}}
\newcommand{\Writeit}{\mbox{\it write}}

\newcommand{\Val}{\mbox{\it val}}

\newcommand{\Event}{\mbox{\it Event}}

\newcommand{\Data}{\mbox{\textit{Data}}}

\newcommand{\Writeon}{\mbox{Write-on-}}
\newcommand{\Readof}{\mbox{Read-of-}}
\newcommand{\Return}{\mbox{Return}}
\newcommand{\Assignmentto}{\mbox{Assign-to-}}

\newcommand{\StateVar}{\mbox{StateVar}}
\newcommand{\Type}{\mbox{Type}}
\newcommand{\calM}{\cal M}

\begin{document}
\title{Kishon's poker game}
\author{Uri Abraham\\ Ben-Gurion University of the Negev}

\maketitle

\begin{abstract}
We present an approach for proving the correctness of distributed algorithms 
that obviate interleaving of processes' actions. The main part of the 
correctness proof is conducted at a higher abstract level and uses Tarskian 
system executions that combine two separate issues: the specification of 
the serial process that executes  its protocol alone (no concurrency here),
 and the specification of the communication objects (no code here). 
In order to explain this approach a   short algorithm for two concurrent processes
 that we call ``Kishon's Poker''  is introduced 
and is used as a platform where this approach is compared to the standard  
one which is based on the notions of global state, step, and history.
\end{abstract}

\section{Introduction}

It is not easy to describe a new approach to an old problem when the meaning of the new terminology is similar 
but not exactly the same as in the established
publications. The longer way is to develop the required terminology 
and definitions in 
detail, but this takes time and the reader may lost patience before the essence of the new approach and its interest are
obtained.  
Aiming to explain our approach in a relatively short exposition, we shall work here with a simple example with which the differences
between the two terminologies and approaches can be highlighted.
We shall present a very simple algorithm for two concurrent processes,
 which we call ``Kishon's Poker Algorithm,''
and describe how two approaches handle its correctness proof:
 the standard\footnote{When we refer here to the standard approach 
it is just for convenience--we do not claim that the existing diversity can be reduced to a single one.} and 
our proposed one. While this doesn't do justice to neither approaches, it
gives a leisurely discussion and a fairly good idea on their different
 merits. 

It should be said already at the beginning that our aim is not
to propose a better approach, but rather to enlarge the existing range of tools by an addition which may be valuable in certain
circumstances.
 
\section{Kishon's algorithm}
\label{KA} The writer and humorist
Ephraim Kishon is no longer popular as he was in my childhood and probably most of you have not read his sketch 
``poker yehudi'' (Jewish Poker). Kishon meets his friend Arbinka who invented a game in which each of the two players has to think about a number
and the one who comes up with the greater number wins. At first, Kishon looses, but when he understands that he must bluff in order
to win he gets his revenge. Senseless as it is, the following algorithm  is motivated by this game. The two players are the processes
$p_0$ and $p_1$, and they run concurrently their protocols (Figure \ref{JP}) just once.

\begin{figure}
\begin{tabular}{|l|l|}\hline
\begin{minipage}[t]{72mm}
 $p_0$'s protocol\\

\begin{enumerate}
\item[$1_0$] $n_0 :=$ pick-a-number$(\,)$;
\item[$2_0$] $R_0 := n_0$;
\item[$3_0$] $v_0:= R_1$;
\item[$4_0$]  \IF\  $(v_0=0\vee v_0=n_0)$ \then\  
  $\> \Val_0 :=0$  \\
\elseif\ $v_0<n_0$ \then\ $val_0 :=1$ \\
  \ELSE\ \hspace{1.5mm} $val_0 := -1$;
\item[$5_0$] \END.
\end{enumerate}
 
\end{minipage}
&
\begin{minipage}[t]{72mm}

$p_1$'s protocol\\

\begin{enumerate}
\item[$1_1$] $n_1 :=$ pick-a-number$(\,)$;
\item[$2_1$] $R_1 := n_1$;
\item[$3_1$] $v_1:= R_0$;
\item[$4_1$]  \IF\ $(v_1=0\vee v_1=n_1)$ \then\ $val_1 :=0$\\
\elseif\ $v_1<n_1$ \then\ $val_1 :=1$ \\
  \ELSE\ $val_1 := -1$;
\item[$5_1$] \END.
\end{enumerate}

\end{minipage}\\
\hline
  \end{tabular}
\caption{The Kishon's Poker algorithm. Registers $R_0$ and $R_1$ are initially $0$. Procedure pick-a-number returns a natural
number $>0$.}
\label{JP}
\end{figure}

The processes communicate with two registers, $R_0$ and $R_1$, written by $p_0$ and $p_1$ respectively (and read by the other process).
These registers carry natural numbers and their initial value is $0$. We assume first, for simplicity, that the registers
are serial, and later on we shall deal with  regular registers (which are even more interesting from our point of view, 
see section \ref{S3.1}).

The local variables of process $p_i$ ($i=0,1$) are $n_i$ and $v_i$ which carry natural numbers and have initial value $0$, and the variable
$val_i$ which carries a value in $\{-1,0,1\}$ and is initially $0$.

We assume a procedure ``pick-a-number'' which returns a (randomly chosen) natural number greater than $0$.

If $E$ is an execution of the protocol by $p_i$ ($i=0,1$) then $E$ consists of four event, $E_1,\ldots,E_4$ which are executions of
the instructions at lines $1_i,\ldots,4_i$ of $p_i$'s protocol. Thus, for example, if $E$ is a protocol execution by $p_0$, then
$E_1$ is the event of invoking pick-a-number and assigning the returned non-zero natural number to variable $n_0$.
 $E_2$ is the write event on register $R_0$ that corresponds to line 2, $E_3$ is the read
event of register $R_1$ and the assignment of the value obtained to variable $v_0$.
  Finally, $E_4$ is the event of assigning a numeric value to variable $val_0$.
  
	If $x$ is any of the program variables or the register $R_i$ of $p_i$, then there is exactly one instruction that can change the initial value of $x$,
   and we denote with $x(E)$ the value of $x$ after it has been finally determined in $E$. For example, if $E$ is a protocol execution by $p_0$, then
   $v_0(E)$ is the value of variable $v_0$ after event $E_3$. That is, $v_0(E)$ is the value (assigned to $v_0$) obtained by reading register $R_1$.

   Suppose for example that $E$ is a protocol execution by $p_0$. Concerning $val_0(E)$ there are four distinct  possibilities:
\begin{enumerate}
 \item If $v_0(E)=0$, then $val_0(E)=0$. Otherwise, $v_0(E)$ is compared to $n_0(E)$ and we have the following possibilities:

     \item If $v_0(E)=n_0(E)$ then $val_0(E)=0$.
     \item  If $v_0(E)<n_0(E)$
then $val_0(E)=1$ (indicating that $p_0$ believes to win the Kishon's game),
\item but if $v_0(E)> n_0(E)$ then $p_0$ sets $val_0(E)=-1$ (admitting it losts the play).
\end{enumerate}

The idea in the way that $p_0$ determines $val_0$ is simple. In case $v_0(E)=0$, $p_0$ ``knows'' that it has read the
initial value of register $R_1$ and not the value of $p_1$'s pick-a-number execution (which must be $>0$). In this case, $p_0$
has no way of knowing whether it is a winner of the game or not, and it sets $val_0$ to the neutral value $0$. In the three remaining
cases, $p_0$ has in its hands the two pick-a-number values, of $p_1$ and its own, and can decide whether $val_0$ should be $+1$ or $-1$
(or $0$ in case of equality).

We assume that the Kishon's poker game is played just once, and our aim is to prove that 
the process that obtains the higher number in its pick-a-number wins the play (i.e. gets a higher $\Val_i$ value than the other player). We will prove the following theorem twice, and thus exemplify in a very simple setting two approaches to the problem of
proving the correctness of distributed algorithms.

\begin{theorem}
\label{T1}
Suppose $E$ and $F$ are concurrent protocol executions by $p_0$ and $p_1$ respectively of their algorithms. If
$n_0(E)<n_1(F)$, then $val_0(E)<val_1(F)$. And symmetrically if $n_1(F)<n_0(E)$ then $val_1(F)<val_0(E)$. If $n_0(E)=n_1(F)$ then $val_0(E)=val_1(F)=0$.
\end{theorem}

In this section we argue informally for this theorem, and to simplify the proof we assume for the moment that the registers are serial.
This assumption will be relaxed in sections \ref{SRR} and \ref{S4}  which deal with regular registers. 
Register $R$ is serial if:
\begin{enumerate}
\item the read/write events on that register are linearly ordered, 
and 
\item the value of any read $r$ of $R$ is equal to the value of the last
write event $w$ on $R$ that precedes $r$ or to the initial value of that register in case there is no write event
that precedes $r$.
\end{enumerate}
 In section \ref{S3} we describe the invariant based
 approach to the proof of this theorem (this is the standard approach), and then in section \ref{Sec6} 
 we reprove this theorem in our model theoretic formulation and approach.

\subsubsection*{Informal argument for Theorem \ref{T1}.}
\label{SSInf}
 Let $E$ and $F$ be protocol execution by $p_0$ and $p_1$ (respectively) which are assumed to be executed concurrently. 
Recal our convention to denote with $x(E)$ (and $x(F)$) the values of any variable $x$ of $p_0$ (respectively $p_1$) at the
end of $E$ ($F$).

  Assume for example
that \begin{equation} 
\label{E1}
n_0(E)<n_1(F).
\end{equation} That is, assume that the number obtained by $p_1$ in its pick-a-number
 is greater than the one obtained by $p_0$. We have to show that
$val_0(E)<val_1(F)$ in this case. The other items of theorem \ref{T1} will follow by a similar argument.

We denote with $E_1, E_2, E_3, E_4$ the events in $E$ (executions of the four instructions at lines $1_0,\ldots,4_0$ of the protocol of $p_0$) and with
$F_1, \cdots, F_4$ the corresponding events in $F$. For any two events $a$ and $b$ we write $a<b$ to say that $a$ precedes $b$ (i.e. $a$
terminates before $b$ begins)\footnote{So $<$ is used for both the number ordering and the temporal precedence ordering. In simple situations such as ours this should create no confusion.}. Thus we have   $E_1<  E_2 < E_3 < E_4$
and $F_1<F_2<F_3<F_4$, but any interleaving of the $E$ events and the $F$ events is possible.

The key fact from which our theorem follows is that each process first writes and then reads. So $E_2<E_3$, i.e. the execution of
$R_0:= n_0$ precedes that of $v_0:= R_1$, and likewise $F_2<F_3$.

There is just one write event on register $R_0$, namely the write $E_2$ which 
is of value $n_0(E)$. 
Likewise there is just one write event on register $R_1$:  the write $F_2$ of value $n_1(F)$.

Comparing the write $E_2$ on register $R_0$ with the read $F_3$ of that register by $p_1$ there are two possibilities.
\begin{enumerate}
\item $F_3<E_2$. In this case  $F_3$ obtains the initial value $0$, and then $v_1(F)=0$ which implies that $val_1(F)=0$.
But as we have that $F_2<F_3<E_2<E_3$, we get that $F_2<E_3$, and hence the read $E_3$ of register $R_1$ obtains the value
of the write $F_2$ on that register which is of value $n_1(F)$. Since we assume that $n_0(E)<n_1(F)$, it follows that $val_0(E)=-1$. Thus we have in this
case that $val_0(E)<val_1(F)$.

\item $E_2<F_3$. In this case $F_3$ obtains the value $n_0(E)$ that was written by $E_2$ on $R_0$, and since we assume that
$n_0(E)<n_1(F)$, $val_1(F)=1$. Depending on the temporal relation between the write $F_2$ and the read $E_3$,
 the read event $E_3$ can either return the initial value $0$ of register $R_1$ (in which case $val_0(E)=0$)
or else the value of the write $F_2$ (in which case $val_0(E)=-1$). In both cases we get that
$val_0(E)<val_1(F)$.
\end{enumerate}

I hope that my readers find this proof of Theorem \ref{T1} satisfactory, and so they may ask ``In what sense that proof is informal, and what is meant here by the term formal proof''? This question is  discussed in section \ref{S3.1}, and at this stage let me only make the following remark. The problem with the proof outlined here is not so much 
in that some details are missing or some assumptions are hidden, but rather in that the connection
between the algorithm text and the resulting executions is not established. In other words,
the proof provides no answer to the question of defining the structures that represent possible executions of the algorithm, 
and hence the very base of the proof is missing.

The standard state-and-history approach gives a satisfactory answer to this question by defining what are executions (also called histories, i.e. sequences of states) and we describe in the following section this standard approach and find
 an invariant that proves theorem \ref{T1}.
Then, in the later sections we outline a formal, mathematical proof of theorem \ref{T1} that follows the steps of the informal
description that was given above. That proof is necessarily longer because it requires some background preparation, namely the
definition of system-executions as Tarskian structures\footnote{The term ``structure'' is so overloaded that we have to be more 
specific, and with {\em Tarskian structures} we refer here to those structures that are used in model theory
as interpretations of first-order languages (see also Section \ref{SReg}).}. As the reader will notice, that proof with Tarskian system executions has the same structure as the informal argument
  that was given above, and this similarity
 between the informal argument and its formal development is in my opinion an important advantage  of the event-based model theoretic approach that we sketch in this paper.

\section{State and history proof of Kishon's algorithm}
\label{S3}
In this section the term {\em state} refers to the notion of {\em global state}, but  in the
event-based approach, as we shall later see, a different notion of state is employed in which
 only local states of individual serial processes are used.
A (global) state of a distributed system is a description of the system at some instant, as if taken by some global snapshot.
A state  is   a function that assigns values to all the state variables. The state variables of the Kishon's algorithm
 are the program variables ($R_i, n_i,v_i$ and $val_i$ for $i=0,1$) and the two program counters $PC_0$ and $PC_1$ of $p_0$ and $p_1$. 
 If $S$ is a state and $x$ is a state variable, then $S(x)$ is the value
of $x$ at state $S$. $S(PC_0)=3_0$ for example means (intuitively) that $p_0$ is about to execute the instruction at line $3_0$ of its
protocol. A {\em step} is a pair of states
$(S,T)$ that describes an execution of an instruction of the protocol, either the instruction at line $S(PC_0)$ if this is a step by $p_0$
or the instruction at line $S(PC_1)$ if this is a step by $p_1$. And $T$ is the resulting state of that execution.

For example, a pair of states $(S,T)$ is a read-of-register-$R_1$ by $p_0$ (also said to be a $(3_0,4_0)$ step) if the following hold.
\begin{enumerate}
\item $S(PC_0)=3_0$ and $T(PC_0)=4_0$,
\item $T(v_0)=S(R_1)$,
\item for any state variable $x$ other than $PC_0$ and $v_0$, $T(x)=S(x)$.
\end{enumerate}

We leave it to the reader to define all eight classes of steps: $(1_0,2_0),\ldots,(4_0,5_0)$ and $(1_1,2_1),\ldots,(4_1,5_1)$.

An initial state, is a state $S$ for which $S(PC_0)=1_0$, $S(PC_1)=1_1$, and $S(x)=0$ for any other state variable. 
(In particular $S(R_0)=S(R_1)=0$).
A final state is a state $S$ for which $S(PC_0)=5_0$ and $S(PC_1)=5_1$. Note that if $S$ is not a final state then there exists
a state $T$ such that $(S,T)$ is a step.

Each state variable has a type. For example, the type of $v_i$ is the set of natural numbers. We denote with $\tau$ the conjunction
which says that each variable is in its type. Clearly the value of each variable
in the initial state is in its type. A simple observation is that
if $(S,T)$ is any step and for every state variable $x$, $S(x)$ is in the type of $x$, then $T(x)$ is also in that type.

 The basic sentential formulas are equality and comparison of state variables. For example $R_0=v_1$, $PC_0\geq 2$, and
 $n_0>0$ are basic formulas. Statements about the types of variables are also basic. For example, ``$n_0$ is in $\mathbb N$''.
If $S$ is any state and $\varphi$ a sentential formula, then $S\models \varphi$ means that $\varphi$ holds in $S$. For example
$S\models v_0<n_0$ iff $S(v_0)<S(n_0)$. Similarly $S\models PC_0\leq 2_0$ iff $S(PC_0)=1_0,2_0$. 
A sentential formula is obtained from basic (primitive) sentential formulas with logical connectives (conjunction, disjunction, negation and
implication).

A sentential formula $\varphi$ is said to be an invariant (also, an inductive invariant) if the following holds.
\begin{enumerate}
\item If $S$ is the initial state then $S\models\varphi$.
\item For every step $(S,T)$, if $S\models \varphi$ then also $T\models\varphi$.
\end{enumerate}
To prove that $\varphi$ is an invariant one has to prove that it holds in the initial state, and then to prove for every step $(S,T)$
 that either $S\not\models\varphi$ or $T\models \varphi$.

A  {\em history sequence} is a sequence of states $S_0,\ldots,S_k,\ldots$ such that (1) $S_0$ is the initial state, (2)  for every state $S_j$ in the sequence that is not a final state, $S_{j+1}$ exists in the sequence and
$(S_j,S_{j+1})$ is a step by $p_0$ or by $p_1$. In our simple algorithm, every terminating history has
exactly four steps by $p_0$ and four by $p_1$. So a terminating history sequence
for the Kishon's algorithm has the form $S_0,\ldots,S_8$. There are as many histories as there are possible ways
of interleavings of $p_0$ and $p_1$ steps.

The following easy theorem is of prime importance even though its proof is immediate:
\begin{quote}
 If $\varphi$ is an invariant formula, and $S_0,\ldots,S_k$ is any history sequence, then
$S_m\models \varphi$ for every index $0\leq m\leq k$.
\end{quote}

The counterpart of Theorem \ref{T1} in this history context is the following.
\begin{theorem}
\label{T2}
If $S_0,\ldots,S_8$ is a terminating history sequence of the Kishon's algorithm, then for each $i\in\{0,1\}$ \[S_8 \models (n_i<n_{1-i}\to val_i <val_{1-i}).\]
\end{theorem}
To prove this theorem we have to find an invariant $\varphi$ such that for each $i\in\{0,1\}$
 \begin{equation}
 \label{Eq}(\varphi\wedge PC_0=5_0\wedge PC_1=5_1)\to (n_i<n_{1-i}\to val_i <val_{1-i}).\end{equation}
\begin{equation}
 \label{Eqq}(\varphi\wedge PC_0=5_0\wedge PC_1=5_1)\to (n_0=n_{1}\to val_0  = val_{1}).\end{equation}
 The difficult part of the proof (which is a typical difficulty for this type of proofs) is to find the required invariant. Once it is found, the proof that $\varphi$ is indeed an invariant is a routine checking.
The invariant that I have found for this theorem is described now.

Let $\alpha$ be the conjunction of the following five sentences.
\begin{enumerate}

\item[$\alpha1.$] $PC_0\geq 2_0\to n_0>0$.

\item[$\alpha2.$] $PC_0\leq 2_0\to R_0=0$.
\item[$\alpha3.$] $PC_0\geq 3_0\to R_0=n_0$.
\item[$\alpha4.$] $v_0\not= 0\to v_0=R_1$. 
\item[$\alpha5.$] $PC_0=5_0\to $

\hspace*{7pt}$(v_0=0\to val_0=0 )\wedge (v_0=n_0\to val_0=0)\wedge$\\ 
$\hspace*{7pt} (0<v_0<n_0\to val_0=1 )\wedge(v_0>n_0\to val_0=-1)$.

\end{enumerate}

Let $\beta$ the conjunction of the corresponding five sentences.
\begin{enumerate}
\item[$\beta1$.] $PC_1\geq 2_1\to n_1>0$.
\item[$\beta2$.] $PC_1\leq 2_1\to R_1=0$.
\item[$\beta3$.] $PC_1\geq 3_1\to R_1=n_1$.

\item[$\beta4$.] $v_1\not= 0\to v_1= R_0$. 
\item[$\beta5$.] $PC_1=5_1\to $

$\hspace{2mm}(v_1=0\to val_1=0 )\wedge (v_1=n_1\to val_1=0)\wedge$\\  
$\hspace*{2mm} (0<v_1<n_1\to val_1=1 ) \wedge (v_1>n_1\to val_1=-1)$.
\end{enumerate}

Let $\gamma$ be the following sentence.
\[PC_0\geq 4_0 \wedge PC_1\geq 4_1\, \to\,  v_0=R_1 \vee v_1=R_0.\]

Let $\tau$ be the sentence saying that each state variable is in its type.

Our invariant $\varphi$ is the conjunction \[\alpha\wedge\beta\wedge \gamma\wedge \tau.\]

To prove that $\varphi$ is an invariant (of the Kishon's algorithm) we have first to prove that $\varphi$ holds in the initial state,
and then prove for every step $(S,T)$ that if $S\models \varphi$, then $T\models\varphi$ as well.  In proving this implication, one can rely only on the definition of steps:
the program itself cannot be consulted at this stage.

\begin{theorem}  $\varphi$ is an invariant.\end{theorem}

 We will not go over this lengthy (and rather routine)
checking, since what interests us here are the concepts involved with the invariant method rather than the details of the proof.
   Assuming that $\varphi$ is indeed an invariant, we shall conclude the proof of theorem \ref{T2} by proving that equations (\ref{Eq}) 
	and (\ref{Eqq}) hold for $i=0,1$.
  This is obtained immediately from the following lemma.
	
	\begin{lemma}
	\label{L3.2}
	Assume that $S$ is any state such that
 $S\models \varphi\wedge PC_0=5_0 \wedge PC_1=5_1$. Then 
 \begin{equation}
 \label{E2}
 S\models (n_0=n_1\to val_0=val_1 = 0), \text{ and }
 \end{equation}

\begin{equation}
\label{E2a}
S\models n_i<n_{1-i} \to (val_i<val_{1-i})
\end{equation}
for $i=0,1$.
\end{lemma}
\begin{proof}
We prove only (\ref{E2a}) (since the proof of (\ref{E2}) uses similar arguments) and  only for $i=0$.
Assume that $S$ is any state such that
 $S\models \varphi\wedge PC_0=5_0 \wedge PC_1=5_1$. 
Assume also that \[S\models n_0<n_1.\] 
Since $S\models \gamma$, there are two cases in the proof that $S\models val_0<val_1$: $S\models v_0 = R_1$, and $S\models v_1 = R_0$.

{\bf Case} $S\models v_0=R_1$. Since $S\models PC_1\geq 3_1$, $\beta3$ implies that $R_1=n_1$. We assume that $n_0<n_1$, and
hence $n_0<v_0$ in $S$. This implies that $S\models \Val_0=-1$ (by $\alpha5$).   
 The possible values for $\Val_1$ are $-1,0,1$, and so it suffices to exclude the possibility of $\Val_1=-1$ in $S$ in order to
conclude that $S\models \Val_0<\Val_1$. Suppose, for a contradiction, that $\Val_1=-1$ in $S$. This implies (by $\beta5$)
that $v_1>n_1$ in $S$. Hence in particular $v_1\not=0$, and $\beta4$ implies that $v_1=R_0$ in $S$. But $R_0=n_0$ (by $\alpha3$)
and so $n_0=v_1>n_1$ follows in $S$ in contradiction to our assumption above.  

{\bf Case} $S\models v_1=R_0$. By $\alpha3$, $R_0=n_0$, and since $n_0>0$ (by $\alpha1$) we get $0<n_0=R_0=v_1 $. Since $n_0< n_1$ in $S$,
$0<v_1<n_1$ in $S$.
So $\Val_1=1$ in $S$. Again, the possible values of $\Val_0$ are $-1,0,1$, and so it suffices to exclude the possibility of
$\Val_0=1$ in order to conclude that $S\models \Val_0<\Val_1$. Suppose for a contradiction that $\Val_0=1$, and conclude that
$0<v_0<n_0$ in $S$. But $v_0=R_1$ (by $\alpha4$) and $R_1=n_1$ in $S$ (by $\beta3$) and hence $n_1<n_0$ in contradiction
to our assumption.\end{proof}

\subsection{An intermediate discussion}
\label{S3.1}
The term ``formal proof'' has a rather definite  meaning in mathematical logic; namely it refers to some set of deduction
rules and axioms in some formal language that determine which sequences of formulas in that language form a proof.
But, in this article, we intend a less stringent usage of this term. A formal-mathematical proof (as opposed to an intuitive informal
 argument) 
is a rigorous mathematical argument that establishes the truth of some sentence written in a formal language.
We require three prerequisites that an argument has to have in order to be considered  a rigorous
mathematical proof in this sense.
(1) There has to be a formal language in which the theorem is stated and in which the main steps of the proof are formulated. (2) there
is a mathematical definition for the class of structures that
interpret the language and a definition for the satisfaction relation that determines when a sentence
$\varphi$  holds in a structure $M$. (3) A proof of a sentence $\varphi$ has to establish that $\varphi$ holds (is satisfied) in every structure
in the intended manifold of structures. In this sense, the proof of theorem \ref{T2} given here in the state-history approach is a satisfying formal-mathematical proof. Of course, some details were omitted and the definitions were not sufficiently
general (they were tailored to the specific Kishon's algorithm), but the  desired prerequisites are there: the language is the sentential
language, and the relation $S\models\varphi$ for a state $S$ and sentence $\varphi$ is (as the reader knows)
 well-defined. With the notions of steps and history
sequences, the framework for proving theorem \ref{T2} is mathematically satisfying, even if my presentation has left something to be desired.

This is not the case with the informal argument for theorem  \ref{T1}: the language in which this theorem is enunciated is not defined, and more
importantly, the structures to which this theorem refers to are not defined. That is, we did not define those structures that
represent executions of the algorithm. As a consequence of these deficiencies the informal argument presented in \ref{SSInf}
does not establish any formal connection between the algorithm text and its executions.   We will establish such a connection
 in section \ref{SReg}, essentially by transforming a certain type of history
sequences into first-order structures. My aim in doing so is
to show how the informal argument for Kishon's algorithm can be transformed into a rigorous mathematical proof, and thus to exemplify the event-based proof framework for distributed systems. To say it in two words, the
formal language for this framework is first-order predicate language, and the structures are Tarskian structures, i.e. interpretations of that predicate language.

The invariant proof method is not without its own typical problems which should be exposed in order to promote the possible
relevance of alternative or complimentary methods.

\noindent
{\bf The difficulty in finding an invariant.} An advantage of the standard state and history based proofs is not only that they yield satisfying mathematical proofs but also that the required notions
 on which such proofs are based
 (states, steps, history-sequences etc.) are easily defined. Yet, and this is a major problem, these proofs are based on the notion of {\em invariant}, and these
 invariants are notoriously difficult (and sometimes very difficult) to obtain.  It took me almost a day to find the invariant
$\alpha\wedge\beta\wedge\gamma$ (checking that it is indeed an invariant can also be long, but with software tools this would be
 an instantaneous work). The reader who is perhaps more experienced in finding invariants may need less
 than a day, but surely more than the five or ten minutes that it takes to find a convincing intuitive argument for the event-based
 theorem \ref{T1} which relies on the checking of different temporal event-scenarios.

{\bf The state and history approach does not show in a clear way where the assumptions on the communication devices are used.}
The role of an assumption in a mathematical proof is often clarified by either finding a proof with a weaker assumption or else
by showing the necessity of that assumption.
In contrast, Theorem \ref{T1}  has an assumption that the registers are serial, but in the assertion-based proof of that theorem  this assumption somehow disappears. It is hidden, in a sense, in the definition of steps.
History sequences, as defined here, can only have serial registers, and thus the status of register seriality as an assumption is unclear.

{\bf Seriality of the registers is a consequence, a theorem, not an assumption.}
From a different angle, the preceding  point can be explained by considering a history sequence $H$ of some algorithm in which the registers
are assumed to be serial.  Seriality of a register means that for every reading step $r$ of the register in $H$, the value of $r$ is the value of the last
previous write step in $H$ that precedes $r$ (or the initial value of $R$ in case no such write step exists).
Now the seriality of $R$ in history $H$ is a {\em theorem} not an assumption on the history. To prove this theorem, we can survey all steps and realize that
only write steps may change the value of the register in a state, and from this we can deduce that the value of any read step is the value of the last write step
in $H$ that precedes that read. Being a theorem, a consequence, how can we say that seriality is an assumption?

{\bf The problem with proving correctness for regular registers.}
This point can be clarified even further by considering regular and safe registers. We may ask: does Kishon's algorithm retain its properties
with only regular (or safe) registers?  We will see in the following subsection that almost the same informal
argument for theorem \ref{T1} works for regular registers, but it would be  a difficult challenge to prove the correctness
of Kishon's algorithm for regular registers in the state-histories approach. In order to represent a single-writer regular register
$R$, one has to represent each extended event
of process $p$ (such as a read or write event of some register) by a pair of {\em invoke} and {\em respond} actions by $p$ that may
appear in a history sequence with several actions of other processes in between.  Moreover, the writing process has to record two values,
the current value and the previous value, and the reading process has to keep an active bag of possible values so that its return value is  one of these possible values. It is not impossible to represents regular registers in such a way but it is quite complicated and I am not aware of any published invariant-based correctness proof for an algorithm that uses these registers.

 \subsubsection{Regular registers}
 \label{SRR}
The notions of safe, regular, and atomic registers were introduced by Lamport \cite{SE} in order to investigate
the situation where read and write operation executions can be concurrent. Thus, seriality is no longer assumed for these events.
 A (single-writer) register $R$ is regular if there is a specific process that can write on $R$, and any read
 of $R$ (by any process) returns a value $v$ that satisfies the following requirement. If there is no write event that
precedes that read, then $v$ is equal to the initial value of $R$ or to the value of some write event that is concurrent with
the read. If there is a write that precedes the read, then the value of the read is the value of the last write on $R$ that precedes the read, or the value of some
 write on $R$ that is concurrent with the read.

 In the context of regular registers, the precedence ordering $<$ on the events is a partial ordering which is not necessarily linear.
 An important property of that partial ordering is the Russell-Wiener property.

 \begin{equation}
 \label{RW}
 \text{For all events }a,b,c,d, \text{ if }a<b,c<d \text{ and } \neg( c<b), \text{ then } a<d.
 \end{equation}
  The justification for this property is evident when
 we think about interval orderings (further details are in Section \ref{SReg}).

  Let's repeat the intuitive proof of Theorem \ref{T1}  for the Kishon's algorithm, but now for regular registers.

 \noindent{\em Correctness of Kishon's Poker algorithm with regular registers.}
 Let E and F be concurrent protocol execution by $p_0$ and $p_1$ (respectively). Recall (from section \ref{KA}) that $E$ consists of four
  events $E_1,\ldots,E_4$ that correspond to executions of lines $1_0,\ldots,4_0$, and $F$ contains the corresponding four events of $p_1$:
	$F_1,\ldots,F_4$. Assume for example that 
	\begin{equation}
	\label{E7}
	n_0 (E) < n_ 1 (F);
	\end{equation}
	
	The initial values of the registers is $0$, and there is just one write event on register $R_0$, namely the write $E_2$ 
	which executes $R_0 := n_0$ and is of value $n_0(E)>0$. Hence (by the regularity of $R_0$) there are only two values that a
	read of register $R_0$ can return: the initial value $0$ and the value $n_0(E)>0$ of the write event $E_2$. Applying this
	observation to $F_3$ (the read event  of $R_0$ by $p_1$), the possible values of $F_3$ are  $0$ and  $n_0(E)>0$. We observe that
	if $E_2<F_3$, then the value of $F_3$ is necessarily $n_0(E)$.
	
	Likewise there is only one write event on register $R_1$, namely $F_2$ which is of value $n_1(F)>0$, and hence the regularity
	of $R_1$ implies that there are only two values that a read of register $R_1$ can return: the initial value $0$ and the
	value $n_1(F)>0$ of the write event $F_2$. Applying this observation to the read $E_3$, 
	\begin{equation}
	\label{Eq8}
	\text{the possible values of } E_3 \ \text{are $0$ and $n_1(F)$.}
	\end{equation} We observe that if $F_2<E_3$ then the value of $E_3$ is necessarily $n_1(F)$.
	
	Taking into account our assumption (\ref{E7}), there are only two possible values for $\Val_0(E)$:
	$0$ and $-1$ (use (\ref{Eq8})). And if $F_2<E_3$ then $\Val_0(E)$ is $-1$. Likewise, there are only two possible values
	for $\Val_1(F)$: $0$ and $1$, and if $E_2<F_3$ then $\Val_1(F)$ is $1$.
	
	So, in order to prove that $\Val_0(E)<\Val_1(F)$, we have to exclude the possibility that $\Val_0(E)=0 \wedge \Val_1(F)=0$.
	We observed that this possibility is indeed excluded if $F_2<E_3$ or if $E_2<F_3$. Therefore the following lemma establishes
	that $\Val_0(E)<\Val_1(F)$ and  proves theorem \ref{T1} for regular registers.

 \begin{lemma}
For every concurrent protocol executions $E$ and $F$ by $p_0$ and $p_1$, $F_2<E_3 \vee E_2<F_3$. 
\end{lemma}
\begin{proof} Assume that  $\neg (F_2< E_3)$. Then we have the following temporal relations.
\[ E_2<E_3,\ \neg (F_2< E_3),\ F_2<F_3.\]
Hence the Russell--Wiener property implies that $E_2<F_3$ as required. \end{proof}

The second part of our paper, sections \ref{SReg}, \ref{Sec6}, and \ref{S4}, is its main part in which we prove the correctness of
Kishon's Algorithm with regular registers in the event and model theoretic framework. Actually, the main part of the
 proof takes just a couple of pages, but the description of its framework necessitates a redefinition of states as (finite) Tarskian
system executions which takes some place.

\section{Tarskian system executions and regular registers}
\label{SReg} 
In this section we define the notion of Tarskian system execution that are used to explicate regular registers. This notion
relies on the notions of first-order language and interpretation which are defined and explained in any logic textbook.
Here we shall describe these notions in a concise way mainly by following an example.

  A Tarskian structure is an interpretation of a first-order language. A language, $L$, is specified
by listing its symbols, that is its sorts (for we shall use  many-sorted languages), its predicates, function symbols, and individual constants. 

For example,  the
language $L_{R}$ that we design now is used for specifying that $R$ is a single-writer regular register (owned by the writer process $p$). The following symbols are in $L_R$.
\begin{enumerate}
\item There are two sorts \Event\ and \Data.
 
\item The unary predicates are 
  $\Readit_R$ and $\Writeit_R$ and the predicate $p$.
Formula $\Read_R(e)$ for example says that event $e$ is a read event of register $R$, and formula $p(e)$ says that event $e$
is by process $p$ (the writer process). 	
	\item A binary predicate $<$ (called the temporal precedence relation) is defined over the events.
		\item There is one function symbol $\Val$, and one individual constants $d_R$ (called the initial value of register $R$).
		\end{enumerate}
	The language $L$ has an infinite set of variables, and it is also convenient to have separate sets of variables for each type.  
	We use here lower-case letters from the beginning of the alphabet, possibly with indexes, as \Event\ variables,
	and the letters
	$x$ and $y$ with or without indexes are used as general variables. It is convenient when designing $L_R$ to be more specific and
	to determine that predicates $\Readit_R$, $\Writeit_R$, $p$, and $<$ apply to sort \Event, that $\Val$ is a function from sort \Event\ to sort
	\Data, and that the constant $d_R$ is in sort \Data.

	A Tarskian structure $M$ that interprets this language $L$ consists of a set $U = U^M$, the ``universe of $M$'' which is the disjoint
	union of two subsets: $\Event^M$ and $\Data^M$ which are the interpretation of the \Event\ and \Data\ sorts by $M$. Then the unary predicates
	$\Writeit_R$, $\Readit_R$, and $p$ 
	are interpreted as subsets of $\Event^M$. The binary predicate $<$ is interpreted as a set of pairs $<^M\subset \Event^M\times \Event^M$
	(called the temporal precedence relation). $\Val$ is
	interpreted as a function $\Val^M:\Event^M\to \Data^M$. And the (initial value) constant $d_R$ is interpreted as a member
	$d_R^M$ of $\Data^M$.

	An interpretation
	of $L_R$ is a rather arbitrary structure $M$ which does not necessarily correspond to the idea that we 
	have of register behavior. It is by means of $L_R$ formulas that we can impose some discipline on the interpretations, and thus
	define the notion of register regularity as the class of those interpretations of $L_R$ that satisfy the required regularity formulas.

	Formulas of a logical language are
	formed by means of its predicates, function symbols and constants, together with the logical symbols: $\forall$ and $\exists$ (these are the quantifiers), and 
	 $\wedge$, $\vee$, $\neg$, $\to$ (these are the connectives). 
	
	Here are some examples of formulas.  $\forall x (\Event(x)\to \Writeit(x)\vee \Readit(x))$  says that every event is a \Writeit\ event or a \Readit\ event. Recalling that variable $a$ is restricted to events, we can write this as $\forall a(\Writeit(a)\vee \Readit(a))$. (This sentence doesn't exclude the possibility that an event is both a read and
	a write, and we can add $\neg\exists a (\Writeit(a)\wedge\Readit(a))$ to do just that). A formula  with no free variables is said to be a sentence, and any sentence $\varphi$ may be true in $M$, $M\models \varphi$, or false $M\not\models \varphi$. 
We can say that
$<$ is a partial ordering on sort \Event\ (i.e. an irreflexive and transitive relation) as follows.
\[ \forall a (a\not < a)\wedge \forall a,b,c(a<b\wedge b<c\to a<c).\]
Here $a\not < a$ is another way to write $\neg(a<a)$.
If $a$ and $b$ are incomparable events in the $<$ temporal precedence relation then we say that $a$ and $b$ are concurrent.
That is, the formula ``$a$ is concurrent with $b$'' is $\neg(a<b\vee b<a)$. (Thus any event is concurrent with itself.)

\begin{definition}
\label{TD}
Let $M$ be an interpretation of our language $L_R$ (or a similarly defined language).
 We say that $M$ is a ``Tarskian system-execution
interpretation of $L_R$''
if \end{definition} 
\begin{enumerate}
\item $<^M$ is a partial ordering that satisfies the Russell-Wiener property (see equation \ref{RW}). Equivalently, 
\begin{equation}
\forall a,b,c,d (a<b\wedge (b\text{ is concurrent with }c) \wedge c < d\to a<d).\end{equation}
\item
For every event $e\in E^M$, the set $\{x\in E^M\mid x <^M e\}$ is finite, and likewise
the set of events that are concurrent with $e$ is also finite. (Finiteness is not a first-order property.)
\end{enumerate}
The Russell-Wiener property makes sense if we think about events as entities that lie in time; they have (or can be represented by)
 a nonempty temporal extension. Every event may be
thought of as being represented with a nonempty interval of instants (say of real numbers). Then $a<b$ means (intuitively) that the temporal extension of event $a$ lies completely to the left of $b$'s extension. So events 
$a$ and $b$ are concurrent if and only if their temporal extensions have a common instant. Now the Russell-Wiener formula can be justified as follows.
Let $t$ be any instant that is both in the temporal interval of $b$ and of $c$. Since $a<b$, every instant of $a$ is before $t$, and similarly every instant of $d$ is after $t$, and hence every instant of $a$ is before every instant of $d$, that is $a<d$.

The finiteness property (item 2) can also be justified for the systems that we have in mind.

\begin{definition}
\label{D4.1}
Let $M$ be some Tarskian system execution interpretation of $L_R$. We say that $M$
  models the regularity of the single-writer register $R$ owned by process $p$ iff it satisfies the following sentences.
\end{definition}
\begin{enumerate}
\item Process $p$ is serial and all write events are by $p$. 
\begin{equation}
\begin{split}
  \forall a,b: &  \\
 &p(a)\wedge p(b)\Rightarrow (a\leq b\vee b\leq a)\ \wedge\\
  & \Writeit_R(a)\Rightarrow p(a).
\end{split}
\end{equation}

\item  No event is both a $\Writeit_R$ and a $\Readit_R$ event. 

\item For every $\Readit_R$ event $r$, $\Val(r)$ satisfies the following (not necessarily exclusive) disjunction.
\begin{enumerate}
\item For some $\Writeit_R$ event $w$ such that $\Val(r)=\Val(w)$,  $w$ is concurrent with $r$, or
\item for some $\Writeit_R$ event $w$ such that $\Val(r)=\Val(w)$, $w<r$ and there is no $\Writeit_R$ event $w'$
with $w<w'<r$, or
\item there is no $\Writeit_R$ event $w$ such that $w<r$, and $\Val(r)$ is the initial value $d_R$ of the register.
\end{enumerate}
\end{enumerate}
When this definition of regularity is expressed as a sentence $\rho$ in $L_{R}$, we can meaningfully say
that   a system execution $M$ that interprets $L_R$ is a model of register regularity iff $M\models \rho$.
 It is
often the case that mathematical English is easier for us to read and understand, and hence most 
statements are given here in mathematical English rather than in first-order formulas.

\section{Non-restricted system executions} 
\label{S4}
We explain in this section the notion of non-restricted semantics of  distributed systemד, a simple but essential 
notion for our approach. 

 Suppose a distributed system composed of 
serial processes that use diverse means of communication devices (such as shared memory registers, message passing channels etc.) where the operations of
each process are directed by some code that the process executes (serially). We ask: what can be said about the semantics of this code
when nothing can be said about its communication devices. On the positive side, we can say for example that instruction at line $k+1$
is executed after the one at line $k$, unless that one is transferring control to another line, we know how to execute
$\IF --- \THEN ---$ instructions etc. And on the negative side we cannot determine the value of a read instruction or the
effect of a write instruction. By {\em non-restricted} we mean the semantics of such a system
 when there is
absolutely no restriction on the behavior of its communication devices.
 So a read of a memory location
(a register) is not required to return the value of the last write event, and it may return any value in the type of that register--even
a value that has never been written. Likewise the contents of messages received is arbitrary, and a message received may have possibly never been sent. Surely not much can
be said about this non-restricted semantics when the different processes are not aware of the presence or absence of the other processes. However, as we shall see, by separating the semantics of the program from
the specification of the communication devices we obtains a greater flexibility and application range in our correctness proofs.

With reference to the code of Figure \ref{JP}, we ask about the semantics of the code for $p_i$ when only the minimal necessary assumptions are made on the read and write instruction executions. 
 These minimal assumptions suppose that
instructions have names and their actions have values. However, and here is the the idea of non-restriction, nothing
relates the value of a communication action to the value of another communication action.    
For example, instruction $R_0:= n_0$ is called a ``write instruction on register $R_0$'' and the  value of its execution is the
value of variable $n_0$,  
and likewise, instruction $v_0:= R_1$ is called ``read of register $R_1$'' and its execution has a value that is assigned to variable $v_0$.
The only thing that can be said about the value of such a read event is that its value  is in the type of $v_0$, i.e. 
natural numbers in our case. The minimal assumptions, however, 
 do not relate this value to
	the value of register $R_1$. In fact we even do not assume that there exists an object called $R_1$, and the expression
	``the value of register
$R_1$''  is meaningless here. The semantics of the program under such minimal assumptions is the
non-restricted semantics of the algorithm, and a detailed definition is given next.

\subsection{An approximation to the non-restricted semantics}
\label{Sub5.1}
In order to give a better explanation of the non-restricted semantics of Kishon's Poker algorithm, we define here that
semantics with local states, local steps, and local history sequences (as opposed to
the global states and steps that were discussed in section \ref{S3}).  This explanation  is only an approximation, a presentation
of the idea, and a fuller presentation will be described in Section \ref{PNP} only after the benefit of non-restricted semantics 
is made evident with the proof of Theorem \ref{T3.3}.  So we begin this section with the notion of local variables and local states and steps.

{\em The local-state variables of process $p_i$ (for $i=0,1$) and their types are the following:} $n_i,v_0,\Val_i$ are of type 
$\mathbb N$ (natural numbers),
and $PC_i$ is of type $\{1,\ldots,5\}$. So registers $R_0$ and $R_1$ are not among
the local state variables. A local state is a function that gives values to the local state variables  in their types.
In the initial local state $S_0$ of $p_i$ we have that $S_0(PC_i)=1$ and $S_0(x)=\bot$ is the undefined value for any other local variable.
 
A non-restricted step by $p_i$ is a pair of $p_i$ local states, and as before we have $(1_0,2_0),\ldots, (4_0,5_0)$ local steps by $p_0$, and
$(1_1,2_1),\ldots,(4_1,5_1)$ local steps by $p_1$. The definition of the read and write local steps however is different from those
of Section \ref{S3}, and as an example we look at $p_0$ local read and write steps.
\begin{enumerate}
\item
A pair of local $p_0$ states $(S,T)$ is denoted ``write-on-register-$R_0$'' (also said to be a local $(2_0,3_0)$ step by $p_0$) when
\begin{enumerate}
\item $S(PC_0)=2$, $T(PC_0)=3$, and
\item all local variables other than $PC_0$ have the same value in $T$ as in $S$.
\end{enumerate}
We define the value of this step to be $S(n_0)$. This is something new in relation to Section \ref{S3}, that steps have values.
Registers do not exist and have no value, but steps do have values in the non-restricted framework. Registers appear in the {\em name} of the step (and this is important, they will serve as predicates)
but they do not record values. 
\item
A pair of $p_0$ local states $(S,T)$ is denoted ``read-of-register-$R_1$'' (also said to be a local $(3_0,4_0)$ step by $p_0$)
when
\begin{enumerate}
\item $S(PC_0)=3$, $T(PC_0)=4$, $T(v_0 )\in \mathbb{N}$,  
\item all local variables other than $PC_0$ and $v_0$ have the same value in $T$ as in $S$.
\end{enumerate}
We define the value of this step to be $T(v_0)$. There is no restriction on the value obtained in a read step (except that it has to be
 in the right type). 

\end{enumerate}

Finally, a non-restricted local
history of $p_i$  is a sequence of local $p_i$ states $S_0,\ldots$, beginning with the initial state $S_0$,
 such that every pair $(S_m,S_{m+1})$ in the sequence is a local step by $p_i$.
A non-restricted history sequence is a good approximation to what we mean here by non-restricted semantics (of the Kishon's Poker algorithm 
of $p_i$),
but the ``real'' definition is given in \ref{Def5.2} in the form of a class of 
system executions whose language $L^i_{NR}$ (for $i=0,1$) is defined first.

\begin{definition} \label{DefLK}
The language $L^i_{NR}$ is a two-sorted language that contains the following features.
\end{definition}
\begin{enumerate}
\item There are two sorts: \Event\ and $\Data$. ($\Data$ has a fixed interpretation as the set ${\mathbb N}\cup\{-1\}$).


\item Unary predicates on \Event\ are: $p_i$, $\Assignmentto n_i$, $\Writeon R_i$, $\Readof R_{1-i}$, and
  $\Return_i$. 

\item
There are two binary 
predicates denoted both
with the symbol $<$, one is the temporal precedence relation on the \Event\ sort, and the other is the ordering relation on the
natural numbers.
\item There is a function symbol $\Val:\Event\to \Data$.
\end{enumerate}
The language contains logical variables with which formulas and sentences can be formed. We reserve lower-case letters such as $a,b,c$
as variables over the \Event\ sort.

We are ready for the definition of non-restricted semantics of Kishon's Poker algorithm.
\begin{definition}
\label{Def5.2}
A system execution $M$ that interprets
$L^i_{NR}$ is said to be a non-restricted 
execution of Kishon's Poker algorithm if it satisfies the properties enumerated in Figure \ref{P1}.
\end{definition}

\begin{figure}
\fbox{
\begin{minipage}[t]{\columnwidth}

\begin{enumerate}

\item[0.] For every event $e$, $p_i(e)$.
There are exactly four events (of $p_i$) and they are linearly ordered by the temporal precedence relation $<$: $a_1<a_2<a_3<a_4$.

\item[1.]
The first $p_i$ event, $a_1$, is an \Assignmentto$n_i$ event. $\Val(a_1)$ is a non-zero natural number.

\item[2.]
The second $p_i$ event, $a_2$, is a \Writeon$R_i$ and $\Val(a_2)=\Val(a_1)$. No other
event is a $\Writeon R_i$ event.

\item[3.] 
The third $p_i$ event, $a_3$, is a \Readof$R_{1-i}$ event. $\Val(a_3)\in {\mathbb{N}}$
is a natural number. No other event is predicated \Readof$R_{1-i}$.

\item[4.] 
The fourth $p_i$ event, $a_4$, is a \Return\ event. As for
the value that $p_i$ returns we have the following:
\begin{enumerate}
\item If $\Val(a_3)=0$ or $\Val(a_3)=\Val(a_1)$ then $\Val(a_4)=0$,
\item If $0<\Val(a_3)<\Val(a_1)$, then $\Val(a_4)=1$,
\item If $\Val(a_3)>\Val(a_1)$, then $\Val(a_4)=-1$.
\end{enumerate}

\end{enumerate}

\end{minipage}
}

\caption{Properties of non-restricted  executions of Kishon's protocol of $p_i$ expressed in the $L^i_{NR}$ language.
We say that $\Val(a_4)$ is the value returned by $p_i$.}
\label{P1}
\end{figure}

In figure  \ref{P1} we collected the little that can be said about non-restricted executions by $p_i$. 
The $p_i$ properties (for $i=0,1$) describe the ways $p_i$ plays the Kishon's game: first an arbitrary nonzero number is picked 
and assigned to $n_i$. This action $a_1$ is predicated by $p_i$ and the $\Assignmentto n_i$ predicates. Then comes event $a_2$ predicated under $p_i$ and 
  \Writeon$R_i$ and whose value is 
equal to the value of the first event. The third event $a_3$ is predicated by $p_i$ and \Readof$R_{1-i}$ and its value (a natural number) is
 unspecified. Finally
the result of the game is calculated, this is $a_4$ the fourth event of $p_i$ whose value depends only on the values of the previous
events by $p_i$.
 
Although these properties of Figure \ref{P1} are written in English, they should formally be expressed in the first-order language
 $L^i_{NR}$ defined above. Only with a formal language  the satisfaction relation that a 
 given structure  satisfies statements in its language is well-defined.

Note in particular item 3 in that list. It says for $i=0$ that
``The third $p_0$ event, $a_3$, is a \Readof$R_1$ event'' and that ``$\Val(a_3)\in {\mathbb{N}}$
is a natural number''. That is, under the minimality condition,  $\Val(a_3)$ is unrelated to any write event,
 and the only requirement is that it has the correct type.
Note also that none of the statements of Figure \ref{P1} relates events of $p_0$ to events of $p_1$.
No concurrency
is involved in the non-restricted specification. For that reason, it is expected that the proof that every non-restricted execution
of Kishon's Poker algorithm satisfies these statements of Figure \ref{P1} would be quite simple. Well, the proof is indeed simple,
but it requires some preliminary definitions which take time, and in order not to delay too much the correctness proof
 we postpone  these preliminary definitions to Section \ref{PNP} and proceed directly to that proof. 
  So, in reading the following section, the reader is asked to rely on intuition and to accept that even if the registers are
		regular, 
the minimal properties of Figure \ref{P1} hold in every execution by $p_i$.
With this assumption we shall prove in the following section the correctness of the Poker algorithm.

 \section{Correctness of Kishon's Poker algorithm for regular registers}
\label{Sec6}

A {\em restricted} semantics is obtained as the conjunction of non-restricted semantics and a specification of the communication
devices. So ``restricted'' means imposing restrictions on the communication devices. In this section the restrictions on the registers
are that they are regular. We shall define in details the restricted semantics and prove Theorem \ref{T3.3} which says essentially that even
with regular registers the correctness condition holds.

 Define the language $L_{NR}$ as the union of the languages $L^0_{NR}$ and $L^1_{NR}$ defined in the previous section and employed in
Figure \ref{P1}.
That is, the symbols of $L_{NR}$ are the symbols of $L^0_{NR}$ and of $L^1_{NR}$. So $L_{NR}$ is a two sorted language, with sorts
\Event\ and \Data. The predicates of $L_{NR}$ are $p_0$, $p_1$ etc. (and the symbols such as $<$ that appear in both languages).

If $M$ is a structure that interprets $L_{NR}$ then the {\em reduct} of $M$ to $p_0$ is the structure $M^0$ defined as follows.
The universe of $M^0$ consists of the set of events $p_0^M$ (that is all the events of $M$ that fall under predicate $p_0$)
and $\Data^{M_0}=\Data^M={\mathbb N}\cup \{-1\}$. All other predicates of $L^0_{NR}$  have the same
interpretation in $M_0$ as in $M$, and the function $\Val^{M_0}$ is the restriction of $\Val^M$ to the set of $p_0$ events of $M$.  In a similar way the restriction of $M$ to $p_1$ is defined. 

 The following is the definition of the class of structures that model executions of the Kishon's Poker algorithm under the assumption
that the registers are regular.

\begin{definition}
\label{D5.6}
A restricted system-execution of  Kishon's Poker algorithm for regular registers is a system execution interpretation of $L_{NR}$, $M$,
 that satisfies the
following two conditions.
\begin{enumerate}
\item Every event in $M$ is either in $p_0$ or in $p_1$. That is, $M\models \forall e(p_0(e)\vee p_1(e))\wedge \neg\exists e(p_0(e)\wedge p_1(e))$. 
\item For every $i=0,1$, $M^i$ satisfies the non-restricted properties of Figure \ref{P1}. ($M^i$ is the reduct of $M$ to $p_i$.)

	\item The two (distinct) registers $R_0$ and $R_1$ are single writer regular registers of $p_0$ and $p_1$ (Definition \ref{D4.1}).
   \end{enumerate}

\end{definition}

Recall (Section  \ref{SReg}) that being a system-execution is already a certain restriction on the structure $M$. For example,
it is required that the 
 $<^M$ relation  satisfies the Russell-Wiener property (equation \ref{RW}).
It follows from this definition that any  restricted system-execution of  Kishon's  algorithm for regular registers
contains eight events: the four $p_0$ event $a_1,\ldots,a_4$ and the four $p_1$ events $b_1,\ldots,b_4$.

We now prove the correctness of the Kishon's Poker algorithm as an absract statement about a class of Tarskian system executions.
 The
 following theorem corresponds to theorem \ref{T1}. 

 \begin{theorem}
 \label{T3.3}
 Assume that $M$ is a system execution of the Kishon's Poker algorithm with regular registers as in Definition \ref{D5.6}. Then
 the following hold in $M$.
 Suppose that $a_1$ and $b_1$ are the Assignment events of $p_0$ and $p_1$ respectively, and $a_4$, $b_4$ are their Return events.
  
	\begin{enumerate}
	\item If $\Val(a_1)<\Val(b_1)$, then $\Val(a_4)<\Val(b_4)$.
	
	\item If $\Val(a_1) > \Val(b_1)$, then $\Val(a_4) > \Val(b_4)$.
	
	\item If $\Val(a_1) = \Val(b_1)$, then $\Val(a_4) = \Val(b_4)=0$.
	\end{enumerate}
	
  \end{theorem}
  \begin{proof} We only prove the first item, since the second has a symmetric proof and the third a similar proof. 
	   The following two lemmas are used for the proof of this theorem. These lemmas rely on the abstract
	 properties of Figure \ref{P1} and on
   the assumed regularity of the registers.

  \begin{lemma}
  \label{ML}
  It is not the case that $\Val(a_3)=\Val(b_3)=0$.
  \end{lemma}
  There are two cases in the proof of the lemma.
  \begin{enumerate}
  \item[Case 1:] $a_2< b_3$. Since $R_0$ is regular and $a_2$ is the only write on $R_0$, and as $b_3$ is a read
		of that register, $\Val(a_2)=\Val(b_3)$.
	But
  $\Val(a_2)=\Val(a_1)>0$ (by items 1 and 2 of Figure \ref{P1}) and hence $\Val(b_3)>0$ in this case as required.

  \item[Case 2:] not Case 1. Then $\neg(a_2<b_3)$. Together with $b_2 < b_3$ and $a_2 < a_3$,
	the Russell-Wiener property implies that 	$b_2 < a_3$.
  Hence $\Val(a_3)=\Val(b_2)$ by the symmetric argument using now the regularity of  register $R_1$. Since $\Val(b_2)>0$, we get
	that $\Val(a_3)>0$ as required.
  \end{enumerate}

 \begin{lemma}
  \label{LM1}
  Suppose that $\Val(a_1)<\Val(b_1)$. 
	\begin{enumerate}
	\item If $\Val(b_3) \not=0$ then 
  \[ \Val(b_4)= 1.\]
 
\item
    Symmetrically, if $\Val(a_3)\not=0$ then
    \[ \Val(a_4)= -1.\]
 \end{enumerate}
    \end{lemma}
 For the proof of this lemma suppose that $\Val(b_3)\not= 0$.
Event $b_3$ is a read of register $R_0$, and there is just one write event on that regular register, namely the write event $a_2$ whose value is the value of the \Assignmentto$n_0$ event $a_1$. The initial value of $R_0$ is $0$. Since
$R_0$ is a regular register, this implies immediately that $\Val(b_3)=0$ or else $\Val(b_3)=\Val(a_2)=\Val(a_1)$. (Indeed,
the value of the read $b_3$ is either the value of some write event, and only  $a_2$ can be that event, or the initial
value of the register.)
Since $\Val(b_3) \not=0$ is assumed,  then \[\Val(b_3)=\Val(a_1)\] follows. The lemma assumes that $\Val(a_1)<\Val(b_1)$,
that is $0<\Val(b_3)<\Val(b_1)$, and hence   $\Val(b_4)= 1$ by property $4(b)$ of $p_1$.
	
The second item of the lemma follows symmetrically, and this ends the proofs of the two lemmas.

 We now conclude the proof of the first item of theorem \ref{T3.3}. Assume that $\Val(a_1)<\Val(b_1)$. We shall prove that 
$\Val(a_4)< \Val(b_4)$. There are two cases.
 \begin{enumerate}
 \item[Case 1] $\Val(a_3)=0$. Then $\Val(a_4)=0$ (by 4(a)). By lemma \ref{ML}, $\Val(b_3)\not= 0$, and hence by lemma \ref{LM1}
 $\Val(b_4)=1$. So $\Val(a_4)<\Val(b_4)$ in this case.

 \item[Case 2] $\Val(a_3)\not=0$. Then $\Val(a_4)=-1$ by lemma \ref{LM1}. Now $\Val(b_4)$ can be $0$, and in this case $\Val(a_4)<\Val(b_4)$ holds.   But if not, if $\Val(b_4) \not = 0$,
then
 $\Val(b_3)\not=0$ by $4_1(a)$. Hence $\Val(b_4)=1$ (again by lemma \ref{LM1}). Thus, in both cases $\Val(a_4)<\Val(b_4)$.

 \end{enumerate}  \end{proof}

\section{Proof of the non-restricted properties}
\label{PNP}

In order to complete the proof of Theorem \ref{T3.3} it remains to prove that the properties of Figure \ref{P1} hold for $p_i$
even when we make no assumptions whatsoever on the registers used. For such a simple protocol with its four instructions that
are executed consecutively, it is evident that every execution generates four actions as described in that figure. It is also evident
that no concurrency is involved in establishing these properties of a single process. However, if we seek a proof  
  that does not rely on our intuitive understanding of the protocol, we need some mathematical framework that relates the code of Figure
\ref{JP}	with the properties of Figure \ref{P1}. That is, we need an explication of the non-restricted semantics with which it is possible to prove that every execution of the code of process $p_i$ results in a structure that satisfies the properties of Figure \ref{P1}.
The development of section \ref{Sub5.1} is not good enough simply because a local history sequence is not a Tarskian structure for which
it is meaningful to say that it satisfies (or not) these properties. In this section we describe a way to present the set of non-restricted executions by process $p_i$ as a set of Tarskian structures.

We first redefine the notion of {\em state} (i.e. a non-restricted state of $p_i$) not as a function that gives values to state variables, but rather as a finite Tarskian
structure. Moreover, that state records not only the attributes of the moment, but actually all the events that led up to that state. That is, a state (in this section) is an extended state structure that includes its own history as well\footnote{In fact, it suffices that the extended state
includes a {\em property} of its history.}. (For earlier work that
tries to  elucidate the notion of state in a similar fashion we refer to \cite{WhatIs} and \cite{OnsystemEx}.)

Let \StateVar\ be some set of {\em state variables} and suppose that every $v\in \StateVar$ has a type $\Type(v)$. Then a state, in
the functional meaning of the word (as is usually defined), is a function $s$ defined over $\StateVar$ such that $s(v)\in \Type(v)$ for every state variable $v$.
Now suppose that $L$ is some first-order language such that every $v\in \StateVar$ is an individual constant in $L$, and its type is 
a sort of $L$ that is
required to be interpreted as $\Type(v)$ in every interpretation of $L$. (Note that $v$ is not a variable of $L$, it is not quantifiable and cannot be a free variable in a formula. It is just a name
of a member of any interpreting structure of $L$.)
Besides these state variables and sorts, $L$ contains other symbols (predicates, function symbols, and constants).
 Let $\calM$ be a class of system execution structures that interpret $L$ and are such that for every
$M\in \calM$ and any
$v\in \StateVar$,  $v^M\in \Type(v)$. The {\em functional state} of $M$, $S=S(M)$, is defined by $S(v)=v^M$ for every
$v\in \StateVar$.
The structures in $\calM$ are said to be {\em extended states.}

The following example will clarify this definition of extended states. 
we define first a logical language $L^0_{K}$ which will be used to define the non-restricted semantics of the Kishon's Poker protocol for
$p_0$. The $L^1_{K}$ language is defined analogously for $p_1$. 
\begin{definition} \label{DefLKK}
The language $L^0_{K}$ is a two-sorted language that contains the following features.
\end{definition}
\begin{enumerate}
\item There are two sorts: \Event\ and $\Data$. ($\Data$ is in this context interpreted as the set ${\mathbb N}\cup\{-1\}$).

\item Individual constants are names of \Data\ values. The following  local variables of $p_0$ are individual constants of $L^0_K$:
 $n_0$, $v_0$, $\Val_0$ and 
$PC_0$. We set $\StateVar= \{n_0,v_0,\Val_0, PC_0\}$. Also, $-1,0,1$ are individual constants with fixed interpretations.

\item Unary predicates on \Event\ are: $p_0$,  $\Assignmentto n_0$, $\Writeon R_0$, $\Readof R_1$, and $\Return_0$. 

\item
There are two binary 
predicates denoted both
with the symbol $<$, one is the temporal precedence relation on the \Event\ sort, and the other is the ordering relation on the
natural numbers.
\item There is a function symbol $\Val:\Event\to \Data$.
\end{enumerate}

Note that   $L^0_{K}$ is richer than the language $L^0_{NR}$ with which the properties of non-restricted executions of $p_0$ were formulated
(in Figure \ref{P1}). In fact, $L^0_K$ is obtained from $L^0_{NR}$ by the additions of the constants in $\StateVar$ and their
sorts.  So any interpreting structure of $L^0_{K}$ is also an interpretation of $L_{NR}$.

Let $\calM$ be the set of all interpretations $M$ of $L_{K_0}$ such that
\begin{enumerate}
\item $\Event^M$ is finite, and $\Data^M={\mathbb N}\cup \{ -1\}$. 
\item  $n_0^M,v_0^M\in {\mathbb N}$ and $PC_0^M\in \{1,2,3,4,5\}$.
\end{enumerate}
Members of $\calM$ are  said to be  extended (non-restricted) states (of $p_0$), and if $M\in \calM$ then the functional
state $S=S(M)$ is defined by $S(x)=x^M$.  

The initial extended state is the structure $M_0\in \calM$ such that $\Event^{M_0}=\emptyset$ (there are no events),
 $PC_0^{M_0} = 1$, the predicates over the events have
empty interpretation (of course) and $x^{M_0}=0$ for any individual constant $x$ (other than $PC_0$).

If $M$ and $N$ are structures for $L^0_K$, then $N$ is said to be an {\em end-extension} of $M$ when $\Event^M$ is an initial
section of $\Event^N$ (in the $<^N$ ordering) and the reduct of $N$ to $\Event^M$ is the structure $M$. 

The $p_0$ extended steps are defined as follows.
\begin{enumerate}
\item[$(1_0,2_0)$ steps] are pairs $(S,T)$ of extended $p_0$ states such that 
\begin{enumerate}
\item $PC_0^S=1_0$, $PC_0^T=2_0$, $n_0^T$ is an arbitrary positive natural number.
\item For some $a_1\not\in \Event^S$ (we say that $a_1$ is a new member) $\Event^T=\Event^S\cup \{a_1\}$, and $S$ is the restriction
of $T$ to $\Event^S\cup {\mathbb N}\cup \{-1\}$.  For every $x\in \Event^S$, $x<^T a_1$. That is, $a_1$ is added after all events
of $S$, and $T$ is an end-extension of $S$.

\item The following hold in $T$. $p_0(a_1)$. $\Assignmentto n_0(a_1)$. $\neg \Writeon R_0(a_1)$. $\neg \Readof R_1(a_1)$. $\neg \Return(a_1)$.
$\Val(a_1)=n_0$.
\end{enumerate}

\item[$(2_0,3_0)$ steps] are pairs $(S,T)$ of extended $p_0$ states such that
\begin{enumerate}
\item $PC_0^S =2_0$, $PC_0^T=3_0$, $n_0^T=n_0^S$.
\item For some new member $a_2$, $\Event^T=\Event^S\cup \{ a_2\}$. $T$ is an end-extension of $S$. 
\item The following holds in $T$. $p_0(a_2)$. $\Writeon R_0(a_2)$. $\neg \Assignmentto n_0(a_2)$.
 $\neg \Readof R_1(a_2)$. $\neg \Return(a_2)$.  $\Val(a_2)= n_0$.

\end{enumerate}

\item[$(3_0,4_0)$ steps] are pairs $(S,T)$ of extended $p_0$ states such that 
\begin{enumerate}
\item $PC_0^S =3_0$, $PC_0^T=4_0$, $n_0^T=n_0^S$, $v_0^T\in {\mathbb N}$ is arbitrary.
\item For some new member $a_3$, $\Event^T=\Event^S\cup \{ a_3\}$. $T$ is an end-extension of $S$. 
\item The following holds in $T$. $p_0(a_3)$.  $\Readof R_1(a_3)$. $\neg \Writeon R_0(a_3)$. $\neg \Assignmentto n_0(a_3)$.
 $\neg \Return(a_3)$.  $\Val(a_3)= v_0$.

\end{enumerate}

\item[$(4_0,5_0)$ steps] are pairs $(S,T)$ of extended $p_0$ states such that

\begin{enumerate}
\item $PC_0^S =4_0$, $PC_0^T=5_0$, $n_0^T=n_0^S$, $v_0^T=v_0^S$. 

\item For some new member $a_4$, $\Event^T=\Event^S\cup \{ a_4\}$. $T$ is an end-extension of $S$. 
\item The following holds in $T$. $p_0(a_4)$.  $\Return(a_4)$.  $\neg \Readof R_1(a_4)$. $\neg \Writeon R_0(a_4)$.
 $\neg \Assignmentto n_0(a_4)$. 
  $\Val_0^T\in \{-1,0,1\}$ is such that $\Val^T(a_4)=\Val^T_0$ and is determined as follows. 
	
	If  $(v^S_0=0\vee v^S_0=n^S_0)$ then  
  $\> \Val^T_0 :=0$  elseif $v^S_0<n^S_0$ then $\Val^T_0 :=1$ else $\Val_0 := -1$.	
	
	Hence the following hold in $T$. If $\Val(a_3)=0$ or $\Val(a_3) =\Val(a_4)$ then $\Val(a_4)=0$.
	If $0< \Val(a_3)<\Val(a_1)$ then $\Val(a_4)= 1$. If $\Val(a_3)>\Val(a_1)$ then $\Val(a_4)=-1$.
\end{enumerate}

\end{enumerate} 

The  $p_1$ steps: 
$(1_1,2_1),\ldots,(4_1,5_1)$ are similarly defined.

A $p_0$ invariant is a (first-order) sentence $\alpha$ in $L_{K_0}$ that holds in the initial extended state $M_0$ and is such that for every
extended step $(S,T)$ of $p_0$, if $S\models \alpha$ then $T\models \alpha$.

Note that if $M$ is any non-restricted extended state such that $PC_i^M = k_i\in\{1,\ldots,4\}$ then there is a state $N$ such that $(M,N)$ is a
$(k,(k+1))$ step. 

\begin{definition}
 A non-restricted extended history sequence of the Kishon's Poker algorithm for $p_i$ ($i=0,1$)  is a sequence of Tarskian non-restricted extended states 
$(M_0,\ldots)$ of $p_i$ such that $M_0$ is the initial state and every pair $(M_i,M_{i+1})$
in the sequence is a step by $p_i$.
\end{definition}
We note that in a maximal non-restricted extended history sequence there are five states.

\begin{definition}
\label{D5.4}
We say that $M$ is a non-restricted execution of the Kishon's protocol for $p_i$ if $M$ is the last state 
in a maximal non-restricted extended history sequence for $p_i$.
\end{definition}
So  if $M$ is a non-restricted execution of the Kishon protocol for $p_i$ then $M\models PC_i=5$.

\begin{theorem}
\label{ThmNR}
If $M$ is  a non-restricted system execution of the Kishon's Poker algorithm, then $M$ satisfies the properties enumerated in
Figure \ref{P1}.
\end{theorem}
\begin{proof}
 Let $M$ be a non-restricted system-execution for $p_0$. So there is
a maximal non-restricted extended history sequence for $p_0$, $M_1,\ldots,M_5$ such that  $M=M_5$ is its last state.
We have to find an invariant $\alpha$ such that
\[ \alpha\wedge PC_0=5 \Rightarrow \text{properties 0--4 of Figure \ref{P1}}.\]
Let $\alpha$ be the conjunction of the properties listed in Figure \ref{P4}. 
  
\begin{figure}[h!]
\fbox{
\begin{minipage}[t]{\columnwidth}
All events fall under predicate $p_0$. 
\begin{enumerate}

\item[1.]
If $PC=1$ then sort \Event\ is empty.

\item[2.] If $PC=2$ then there is just one event, $a_1$, which is such that 
 \Assignmentto$n_0(a_1)$  
and $\Val(a_1)>0$ is a natural number.

\item[3.] If $PC=3$, then there are exactly two events. The first is as in item 2 and the second event, $a_2$, is
such that    $\Writeon R_0(a_2)$ and $\Val(a_2)=\Val(a_1)$.

\item[4.] If $PC=4$, then there are exactly three events. The first and second are as in items 2 and 3, and the third event, $a_3$,
is such that  $\Readof R_1(a_3)$ and $\Val(a_3)\in {\mathbb{N}}$ is a natural number.

\item[5.] If $PC=5$ then there are exactly four events. The first three are as in items 1,2,3, the fourth event, $a_4$, is 
such that $\Return_0(a_4)$ and $\Val(a_4)$ satisfies the following:
\begin{enumerate}
\item If $\Val(a_3)=0$ or $\Val(a_3)= \Val(a_1)$ then $\Val(a_4)=0$,
\item If $0<\Val(a_3)<\Val(a_1)$, then $\Val(a_4)=1$,
\item If $\Val(a_3)>\Val(a_1)$, then $\Val(a_4)=-1$.
\end{enumerate}

\end{enumerate}
\end{minipage}
}
\caption{Inductive properties expressed in the $L^0_{K}$ language which serve in the proof of Theorem \ref{ThmNR}.}
\label{P4}
\end{figure}

It is not difficult to prove that
$\alpha$ is an invariant. This  shows that 
$M = M_5\models \alpha$ and since $M_5\models PC_0=5$ it follows immediately that $M_5$ satisfies the properties enumerated in 
Figure \ref{P1}, and hence Theorem \ref{T3.3} applies. \end{proof}

 \section{Discussion}
\label{S5}
 The Kishon's Poker algorithm is a trivial short algorithm for two processes that execute concurrently. It serves here as a platform to
introduce an event-based {\em model theoretic} approach to the problem of proving the correctness of distributed systems
and to
 compare this approach to the {\em standard}
 approach that is based on the notions of global state, history, and invariance. 
 
Invariance is certainly an extremely important key concept, but there are  situations
in which it is natural to argue about  the events and their temporal  interrelations, whereas finding the invariance and its 
proof is quite difficult.
This observation is well-known, as is the well-documented observation that proofs that rely on these events and temporal interrelations
often lead to grave errors. For example Lamport writes in  \cite{Ass}:
\begin{quote}
Most computer scientists find it natural to reason about a concurrent program in terms of its behavior--the sequence of events generated by its execution. Experience has taught us that such reasoning is not reliable; we have seen too many convincing proofs of incorrect algorithms. This has led to assertional proof methods, in which one reasons about the program's state instead of its behavior.
\end{quote}

The approach outlined in this paper to the problem of proving properties of distributed systems is guided by the desire
to preserve the naturalness of the behavioral approach and to combine it with reliable mathematical rigor. This approach is characterized by the following features.
\begin{enumerate}

\item 
The interleaving semantics and its global states and histories are not used. Instead of global states and interleaving of actions by the different processes, only local states and histories are used in
order to express what each process does in a way that is detached from the properties of the communication devices.
Local states of an
individual process are needed in order to define how the code is executed by a process without any commitment to any specification of
the shared communication devices. 
The resulting structures are called ``non-restricted Tarskian system executions'' and the manifold of all of these
structures is called the ``unrestricted system''. These structures are unrestricted in the sense
that the values of inter-process communication objects of one process are not  connected to values of other processes, simply because the other processes and
the common communication objects are not represented in the local states of a process. 
 The properties of these unrestricted
system executions are properties that refer to each of the processes separately, as if that process lives in a world
by itself. Figure  \ref{P1} is an example of such properties. Note that
 no concurrency is involved so far, and only local analysis is involved.

\item The correct specification of the communication devices is formulated (again in some first-order language)
in a way that is not connected to any specific system of programs. For example, regularity is a property of registers
that is not connected to the Kishon's Poker algorithm or any other algorithm that uses them.
\item  If from the
class of non-restricted system executions we take only those
system executions in which the communication devices work properly (i.e. satisfy the communication device specifications)
then we get the system $\cal M$ that represents the manifold of all possible executions of  the algorithm under the required assumptions (such as regularity)
on the
communication devices. See for example Definition \ref{D5.6}.

\item The correctness of the algorithm is expressed wih a certain condition $\tau$, and to prove it one has to prove that any system 
execution $M$ in $\cal M$ satisfies $\tau$. 

\item The proof of correctness is thus composed of two main stages: establishing properties of the non-restricted
system executions (obtained by analyzing serial processes with their local states and histories), and then using
these properties in conjunction with assumed properties of the communication devices in order to prove the correctness condition $\tau$.
This separation of concerns is a main feature of the event-based approach outlined here, and the following slogan expresses this.
\begin{quote} \em The specification  of a distributed system depends on
properties of sequential programs that work in isolation and on generally formulated properties of the communication devices that the processes employ.
\end{quote}
\end{enumerate}

There is an obvious price to pay for such a model-theoretic framework--it requires a certain (minimal) familiarity with basic notions in
logic and model theory. Some may say that this price is a barrier that cannot be accepted for a framework that claims to be
intuitive and natural. But there is plenty of historical evidence to show that what was once considered as difficult
 becomes with time standard material when 
better ways of explaining and presenting complex issues are developed.   

It is interesting to note that that the event-based correctness proof framework, exemplified here with the
Kishon's Poker algorithm, has its origin, to some extent, in Lamport's earlier work \cite{SE} when he presented system-executions
 (and the notion of
higher-level events). This is my reason to continue to use the term system-execution, to acknowledge the
connection with Lamport's earlier articles.  Two important features however were missing from that earlier work which limited its degree of mathematical formality
and restricted its range of applications. Firstly, those system-executions were not full-fledge Tarskian structures 
in the model-theoretic sense (they were not interpretations of a first-order language), and
secondly the notion of local state was missing from these proofs, and thus Lamport's system-executions were not mathematically
connected to the programs that their processes employ: there remained a gap between the software and the system-executions that represent executions of that software.
A  bridge was missing between these two aspects (code and execution), and  perhaps it is 
because of this
gap that the earlier system-executions were forsaken, and Lamport himself concentrated in his later work on the invariant approach.
 The bridge that I describe here using the Kishon's Poker algorithm as an example is quite simple. It relies on the notion of local states and on  non-restricted semantics of single-process protocols. The advantage of non-restricted semantics is the separation of
the two issues: the execution of the code by each of the individual processes, and the specification
of the communication devices. These ideas were first presented in \cite{A} and developed in \cite{abraham} and in some later publications,
and surely much work is still needed in order to transform them into a well-developed useful framework for proving properties of
distributed systems.

\end{document}